\documentclass[10pt,letterpaper]{amsart}
\usepackage{amsmath,amsthm,amssymb,latexsym} 
\usepackage{fullpage} 
\usepackage{enumerate} 
\usepackage{graphicx} 
\usepackage{color} 
\usepackage{accents} 
\usepackage{import} 
\usepackage[pdftex,bookmarks,colorlinks,breaklinks]{hyperref}  
\definecolor{dullmagenta}{rgb}{0.4,0,0.4}   
\definecolor{darkblue}{rgb}{0,0,0.4}
\hypersetup{linkcolor=red,citecolor=blue,filecolor=dullmagenta,urlcolor=darkblue} 

\title{Factorizations of Rational Matrix Functions with Application to Discrete Isomonodromic Transformations
and Difference Painlev\'e Equations}
\author{Anton Dzhamay \\
School of Mathematical Sciences\\ 
University of Northern Colorado\\
Greeley, CO 80639}

\begin{document}
	
\newtheorem{theorem}{Theorem}
\newtheorem{prop}[theorem]{Proposition}
\newtheorem{corollary}[theorem]{Corollary}
\newtheorem{lemma}[theorem]{Lemma} 
\theoremstyle{definition}
\newtheorem{definition}[theorem]{Definition}
\numberwithin{theorem}{section}
\numberwithin{equation}{section}

\newenvironment{remark}{$\triangleleft$ {\bf
Remark:}}{$\triangleright$\medskip}
\newenvironment{example}{$\triangleleft$ {\bf
Example}}{$\triangleright$\medskip}
\newenvironment{notation}{\noindent$\triangleleft$ {\bf
Notation:}}{$\triangleright$\medskip}
\newenvironment{assumption}{\noindent$\triangleleft$ {\bf
Assumption:}}{$\triangleright$\medskip}

\maketitle
\begin{abstract}
	We study factorizations of rational matrix functions with simple poles on the Riemann sphere. For the quadratic case
	(two poles) we show, using multiplicative representations of such matrix functions, that a good coordinate system 
	on this space is given by a mix of residue eigenvectors
	of the matrix and its inverse. Our approach is motivated by the theory of discrete isomonodromic transformations
	and their relationship with difference Painlev\'e equations. In particular, in these coordinates, basic 
	isomonodromic transformations take the form of the discrete Euler-Lagrange equations. Secondly we show that dPV
	equations, previously obtained in this context by D.~Arinkin and A.~Borodin, can be understood as simple
	relationships between the residues of such matrices and their inverses.
\end{abstract}


\section{Introduction} 
\label{sec:introduction}
In this paper we are concerned with various ways of introducing coordinates on the space of rational
matrix functions $\mathbf{L}(z)$ (with simple poles) on the Riemann sphere. Examples like this play an important role
in various applications, like Yang-Baxter maps and matrix solitons \cite{Ves:2003:YMAID,Veselov:2007qf}, 
Lax equations and isomonodromy transformations on algebraic curves, \cite{Kri:2002:IEACCTWE,Kri:2002:VBLEAC}, 
discrete integrable systems, \cite{MosVes:1991:DVSCISFMP,Sur:2004:DLM},
and others. Our main interest is related to the study of discrete isomonodromic transformations 
and their relationship with discrete Painlev\'e equations. The theory of discrete isomonodromic transformations 
was developed by A.~Borodin for polynomial matrices in \cite{Bor:2004:ITLSDE} and adapted to 
the notion of local monodromy using rational matrix functions by I.~Krichever in \cite{Kri:2004:ATDEWRECRP}.
In \cite{AriBor:2006:MSDDPE} D.~Arinkin and A.~Borodin used the theory of d-connections on vector bundles 
to explained the appearance of difference Painlev\'e equations, considered from the 
geometric point of view of H.~Sakai \cite{Sak:2001:RSAWARSGPE}, in the theory of discrete isomonodromic
transformations and later, in \cite{AriBor:2007:TDITP}, they introduced the notion of the $\tau$-function
of such transformation. These $\tau$-functions appear as the gap probabilities in the discrete
probabilistic models of random matrix type. 

I.~Krichever conjectured that discrete isomonodromic transformations can be written
in the Lagrangian form and that they should be related to the universal symplectic form 
of Krichever-Phong, \cite{KriPho:1998:SFTS}. 
In \cite{Dzhamay:2008yq}, using the methods of 
\cite{MosVes:1991:DVSCISFMP}, we verified this conjecture for the quadratic (two-pole) case
by using the multiplicative coordinates on the space of these matrix functions and finding
the explicit formula for the Lagrangian. Recently F.~Soloviev showed 
that our Lagrangian symplectic form coincides with the reduction of the quadratic
symplectic form of Krichever-Phong to certain symplectic leaves \cite{Sol:2008:QAAIC}.

Unfortunately, the multiplicative coordinates are not easily obtained from other 
characteristic properties of $\mathbf{L}(z)$, such as its residue matrices. This is 
an obstacle to the generalization of our results to the higher-order case.  
In this paper we argue that one way around this difficulty is to consider, in 
addition to the residues of $\mathbf{L}(z)$, the residues of the inverse 
matrix $\mathbf{L}^{-1}(z)$. Then good coordinates on the space of 
$\mathbf{L}(z)$, again in the quadratic case, are given by half of the residues data
of $\mathbf{L}(z)$ and half of the residue data of $\mathbf{L}^{-1}(z)$. 
Such residues also explain the almost symmetric form of the expressions
for the multiplicative coordinates of $\mathbf{L}(z)$
 and allow us to recognize
dPV equations as simple relations between the residues of $\mathbf{L}(z)^{\pm1}$ and the
residues of the transformed matrices $\tilde{\mathbf{L}}(z)^{\pm1}$.

The paper is organized as follows. In Section~\ref{sec:main_object_of_study} we give a short 
overview of the additive representation of $\mathbf{L}(z)$. 
In Section~\ref{sec:multiplicative_form_of_rational_matrix_functions} we study 
the multiplicative representation of $\mathbf{L}(z)$ and establish the 
relationship between the eigenvectors of the residues of $\mathbf{L}(z)^{\pm1}$
and the eigenvectors of the left and right-divisors of $\mathbf{L}(z)$. We also
significantly simplify many of the arguments and formulas of \cite{Dzhamay:2008yq}. 
Finally, in Section~\ref{sec:isomonodromic_transformations_and_dpv} we restrict our attention 
to the rank-two case. It this case it is possible to introduce the so-called spectral
coordinates on the space of $\mathbf{L}(z)$ in such a way that the equations relating
the spectral coordinates of $\mathbf{L}(z)$ and $\tilde{\mathbf{L}}(z)$ are precisely the
difference Painlev\'e equations. Our main observation here is the following. In addition
to spectral coordinates, divisor
of zeroes and poles, and some asymptotic behavior at infinity, 
the entries of the matrix $\mathbf{L}(z)$ also depend on a choice of a gauge
with respect the action of the group of constant non-degenerate diagonal matrices.
Understanding the change of this gauge from $\mathbf{L}(z)$ to $\mathbf{L}^{-1}(z)$ and
$\tilde{\mathbf{L}}(z)$ allows us to easily obtain the expressions of the multiplicative
representation of $\mathbf{L}(z)$ in the spectral coordinates and  the 
dPV equations.

\section{Additive form of Rational Matrix Functions} 
\label{sec:main_object_of_study}

Let $\mathbf{L}(z)$ be a rational matrix function on the Riemann sphere, 
$\operatorname{rank} \mathbf{L}(z) = m$, satisfying the following general conditions.
First, we require that there exists a normalization point $z_{0}$ at which $\mathbf{L}(z)$ is regular, 
$\lim_{z\to z_{0}} \mathbf{L}(z) = \mathbf{L}_{0}$, $\det \mathbf{L}_{0} \neq 0$, and
all eigenvalues of $\mathbf{L}_{0}$ are distinct.
We can then do a gauge transformation
to make $\mathbf{L}_{0}$ diagonal, thus reducing the global gauge group to the group 
of diagonal matrices. Without any loss of generality we can assume that $z_{0}=\infty$,
and so 
\begin{equation}
	\lim_{z\to\infty} \mathbf{L}(z) = \mathbf{L}_{0} = \operatorname{diag}\{\rho_{1},\dots, \rho_{m}\}.
\end{equation}

Second, we impose the following conditions on the pole structure of $\mathbf{L}(z)$ and its inverse
$\mathbf{M}(z) = \mathbf{L}(z)^{-1}$. We require that $\mathbf{L}(z)$ is holomorphic except for  
\emph{simple} poles at the points $z_{1},\dots, z_{k}$, $\mathbf{M}(z)$ is holomorphic except for 
 \emph{simple} poles at the points $\zeta_{1},\dots, \zeta_{k}$,  all $z_{i}$ and $\zeta_{j}$ 
are distinct, and the determinant $\det \mathbf{L}(z)$ has  also only \emph{simple} poles at $z_{i}$ and \emph{simple}
zeroes at $\zeta_{j}$.  These conditions mean that the residues $\mathbf{L}_{i}=\operatorname{res}_{z_{i}}\mathbf{L}(z)$
and $\mathbf{M}_{j}=-\operatorname{res}_{\zeta_{j}}\mathbf{M}(z)$ (the negative sign here is for future convenience) 
are matrices of rank one. Using the $\dag$ symbol to indicate a row vector, 
we have:
\begin{align}
	\mathbf{L}(z) &= \mathbf{L}_{0} + \sum_{i=1}^{k} \frac{\mathbf{L}_{i}}{z-z_{i}},
	\qquad\text{where } \mathbf{L}_{0}=\operatorname{diag}\{\rho_{1},\dots,\rho_{m}\}\text{ and }
	\mathbf{L}_{i} = \mathbf{a}_{i} \mathbf{b}^{\dag}_{i},\label{L(z)-props-a}\\
	\det \mathbf{L}(z) &= \rho_{1}\cdots \rho_{m} 
	\frac{\prod_{i=1}^{k} (z-\zeta_{i})}{\prod_{j=1}^{k} (z-z_{j})},\label{L(z)-props-b}\\
	\mathbf{L}(z)^{-1}= \mathbf{M}(z) &= \mathbf{M}_{0} - \sum_{i=j}^{k}\frac{\mathbf{M}_{j}}{z - \zeta_{j}},
	\qquad\text{where } \mathbf{M}_{0}=\mathbf{L}_{0}^{-1},\qquad
	\mathbf{M}_{j} = \mathbf{c}_{j} \mathbf{d}^{\dag}_{j}\label{L(z)-props-c},
\end{align}
We call the above representations of $\mathbf{L}(z)$ and 
$\mathbf{M}(z)$ \emph{additive representations} and the vectors $\mathbf{a}_{i}$, $\mathbf{b}_{i}^{\dag}$ 
(resp. $\mathbf{c}_{i}$, $\mathbf{d}_{i}^{\dag}$) \emph{additive eigenvectors} of 
$\mathbf{L}(z)$ (resp. $\mathbf{M}(z)$). Note that in terms of $\mathbf{L}(z)$ the vectors $\mathbf{c}_{j}$ 
(resp.~$\mathbf{d}_{j}^{\dag}$) can be characterized as the left (resp.~right) null-vectors of $\mathbf{L}(\zeta_{j})$,
and the similar statement is true for $\mathbf{M}(z_{i})$ and $\mathbf{a}_{i}$, $\mathbf{b}_{i}^{\dag}$.
By the divisor of $\mathbf{L}(z)$ we mean the divisor
$\mathcal{D}$ of its determinant, $\mathcal{D} = \sum_{i} z_{i} - \sum_{i} \zeta_{i}$. We denote the space of
matrices $\mathbf{L}(z)$ satisfying the conditions (\ref{L(z)-props-a})--(\ref{L(z)-props-c}) by 
$\mathcal{M}_{r}^{\mathcal{D}}$.

An important question is how to choose a good coordinate system on the space $\mathcal{M}_{r}^{\mathcal{D}}$.
For example, given $\mathbf{L}_{i}$, we can determine $\mathbf{a}_{i}$ and $\mathbf{b}_{i}^{\dag}$ only up
to a common scaling factor. This factor has to be adjusted to ensure that $\det \mathbf{L}(z)$ has
zeros at $\zeta_{i}$, which a complicated condition on $\operatorname{tr}(\mathbf{L}_{i})$. Same problem
is present for the collection $\{\mathbf{c}_{i}, \mathbf{d}_{i}^{\dag}\}$. Some insight for a good choice
of coordinates is provided by the study of the isospectral and isomonodromic transformations, which suggests
that half of the coordinates should be taken from the residues of $\mathbf{L}(z)$ and half from 
the residues of $\mathbf{M}(z)$. In fact, in the quadratic case, we have the following result.

\begin{theorem}\label{thm:coords} When $\mathbf{L}(z)$ has $k=2$ poles, 
	the vectors $(\mathbf{c}_{2},\mathbf{d}_{1}^{\dag}; \mathbf{a}_{2}, \mathbf{b}_{1}^{\dag})$,
	considered up to rescaling (i.e., as points in $\mathbb{P}^{r-1}$), are coordinates on the space
	$\mathcal{M}_{r}^{\mathcal{D}}$. To recover $\mathbf{L}^{\pm1}(z)$, consider the function
	\begin{align}
		\mathcal{L}((\mathbf{x}_{2},\mathbf{x}_{1}^{\dag}),(\mathbf{y}_{2},\mathbf{y}_{1}^{\dag})) &= 
		(z_{2} - z_{1}) \log(\mathbf{x}_{1}^{\dag} \mathbf{L}_{0} \mathbf{x}_{2}) + 
		(z_{1} - \zeta_{2}) \log(\mathbf{y}_{1}^{\dag} \mathbf{x}_{2}) \notag \\
		&\qquad + (\zeta_{2} - \zeta_{1})\log(\mathbf{y}_{1}^{\dag} \mathbf{L}_{0}^{-1} \mathbf{y}_{2}) +
		(\zeta_{1} - z_{2}) \log(\mathbf{x}_{1}^{\dag} \mathbf{y}_{2}).\label{eq:Lagrangian}
	\end{align} 
	Then
	\begin{align}
		\mathbf{a}_{1} &= -\frac{\partial \mathcal{L}}{\partial \mathbf{x}_{1}^{\dag}} 
			((\mathbf{c}_{2},\mathbf{d}_{1}^{\dag}), (\mathbf{a}_{2}, \mathbf{b}_{1}^{\dag})); & \qquad 
		\mathbf{b}_{2}^{\dag} &= \phantom{-}\frac{\partial \mathcal{L}}{\partial \mathbf{x}_{2}} 
			((\mathbf{c}_{2},\mathbf{d}_{1}^{\dag}), (\mathbf{a}_{2}, \mathbf{b}_{1}^{\dag})); \label{eq:Ls}\\
		\mathbf{c}_{1} &= \phantom{-}\frac{\partial \mathcal{L}}{\partial \mathbf{y}_{1}^{\dag}} 
			((\mathbf{c}_{2},\mathbf{d}_{1}^{\dag}), (\mathbf{a}_{2}, \mathbf{b}_{1}^{\dag})); & \qquad 
		\mathbf{d}_{2}^{\dag} &= -\frac{\partial \mathcal{L}}{\partial \mathbf{y}_{2}} 
			((\mathbf{c}_{2},\mathbf{d}_{1}^{\dag}), (\mathbf{a}_{2}, \mathbf{b}_{1}^{\dag})). \label{eq:Ms}
	\end{align}
\end{theorem} 

The proof of this theorem is based on the multiplicative representations of $\mathbf{L}(z)$, which we consider next.


\section{Multiplicative form of Rational Matrix Functions} 
\label{sec:multiplicative_form_of_rational_matrix_functions}

\subsection{Elementary Divisors} 
\label{ssec:elementary_divisors}
To change from the additive to a multiplicative 
representation, we first need to define the building blocks for it. We call such building blocks the
\emph{elementary divisors}.

\begin{definition}\label{def:el-divs} An \emph{elementary divisor} is a rational $m\times m$-matrix function 
	$\mathbf{B}(z)$ on the Riemann sphere of the form 
		\begin{equation}
		\mathbf{B}(z) = \mathbf{I} + \frac{\mathbf{G}}{z-z_{0}},\qquad\text{where } \mathbf{G} = 
		\mathbf{p} \mathbf{q}^{\dag}\text{ is a matrix of rank one.}
	\end{equation}
\end{definition}
A direct calculation establishes the following elementary facts.

\begin{lemma}\label{lem:el-div-prop}
	Let $B(z)$ be an elementary divisor. Then 
	\begin{enumerate}[(i)]
		\item $\det \mathbf{B}(z) = (z - \zeta_{0})/(z - z_{0})$, where $\zeta_{0} = z_{0} - \mathbf{q}^{\dag} \mathbf{p}$;
		\item $\mathbf{B}(z)^{-1} = \mathbf{I} - \mathbf{G}/(z - \zeta_{0})$.
	\end{enumerate}
\end{lemma}	

In fact, for us it will be more convenient to fix the points $\zeta_{0}$ and $z_{0}$ on $\mathbb{CP}^{1}$. 
Thus, we say that a pair $(\zeta_{0},z_{0})$ corresponds to an elementary divisor $\mathbf{B}(z)$ 
with the determinant $\det\mathbf{B}(z) = (z-\zeta_{0})/(z-z_{0})$. Any generic matrix $\mathbf{B}(z)$ 
with such determinant and normalized by the condition $\mathbf{B}(z)\to \mathbf{I}$ as $z\to\infty$ is of the form
$\mathbf{B}(z) = \mathbf{I} + \mathbf{G}/(z-z_{0})$ with $\mathbf{G}$ of rank one and $\det \mathbf{B}(\zeta_{0})=0$
(more carefully, $\mathbf{B}(z)^{-1}$ should be regular at $z_{0}$ and $\mathbf{B}(z)$ should be regular 
at $\zeta_{0}$). 
Thus, $\mathbf{B}(\zeta_{0})$ has a left null-vector $\mathbf{q}^{\dag}$ and a right null-vector $\mathbf{p}$.
We then immediately get that $\mathbf{G}=\mathbf{p} \mathbf{q}^{\dag}$, where we need to normalize the vectors
$\mathbf{p}$ and $\mathbf{q}^{\dag}$ so that $\mathbf{q}^{\dag} \mathbf{p} = z_{0} - \zeta_{0}$. In a more
invariant form this can be written as
\begin{equation}
	\mathbf{B}(z) = \mathbf{I} + \frac{z_{0} - \zeta_{0}}{z - z_{0}} 
	\frac{\mathbf{p} \mathbf{q}^{\dag}}{\mathbf{q}^{\dag} \mathbf{p}}.
\end{equation}
From that point of view, the formula for the inverse matrix follows from the vanishing of the residue of the identity
$\mathbf{B}(z)\mathbf{B}(z)^{-1} = \mathbf{I} = \mathbf{B}(z)^{-1} \mathbf{B}(z)$ at $z_{0}$,
\begin{equation}
	\mathbf{B}(z)^{-1} = \mathbf{I} + \frac{\zeta_{0} - z_{0}}{z - \zeta_{0}} 
	\frac{\mathbf{p} \mathbf{q}^{\dag}}{\mathbf{q}^{\dag} \mathbf{p}}.
\end{equation} 

The following easy properties of elementary divisors are very useful for what follows.

\begin{lemma}\label{lem:el-div-solve}
	Let 
	\begin{equation}
		\mathbf{B}(z) = \mathbf{I} + \frac{z_{0} - \zeta_{0}}{z - z_{0}} 
		\frac{\mathbf{p} \mathbf{q}^{\dag}}{\mathbf{q}^{\dag} \mathbf{p}}.
	\end{equation}
	Then
	\begin{enumerate}[(i)]
		\item 
		\begin{equation}
			\mathbf{B}(z)\mathbf{p} = \left(\frac{z - \zeta_{0}}{z-z_{0}}\right) \mathbf{p}
			\qquad \text{and} \qquad
			\mathbf{q}^{\dag} \mathbf{B}(z) = \left(\frac{z - \zeta_{0}}{z-z_{0}}\right) \mathbf{q}^{\dag}.
		\end{equation}
		\item Suppose that at some point $z^{*}$ we have $\mathbf{B}(z^{*})\mathbf{w} = \mathbf{v}$
		for some vectors $\mathbf{v}$ and $\mathbf{w}$. Then
		\begin{equation}
			\mathbf{B}(z) = \mathbf{I}  + \frac{1}{z-z_{0}} 
			\left( (z_{0} - z^{*}) \frac{ \mathbf{w} \mathbf{q}^{\dag}}{\mathbf{q}^{\dag} \mathbf{w}} + 
			(z^{*} - \zeta_{0}) \frac{ \mathbf{v} \mathbf{q}^{\dag}}{\mathbf{q}^{\dag} \mathbf{v}} \right)
		\end{equation}
		(i.e., we can determine $\mathbf{p}$ from $\mathbf{v}$, $\mathbf{w}$, and $\mathbf{q}^{\dag}$).
		Similarly, if $\mathbf{w}^{\dag} \mathbf{B}(z^{*}) = \mathbf{v}^{\dag}$,
		\begin{equation}
			\mathbf{B}(z) = \mathbf{I}  + \frac{1}{z-z_{0}} 
			\left( (z_{0} - z^{*}) \frac{ \mathbf{p} \mathbf{w}^{\dag}}{\mathbf{w}^{\dag} \mathbf{p}} + 
			(z^{*} - \zeta_{0}) \frac{ \mathbf{p} \mathbf{v}^{\dag}}{\mathbf{v}^{\dag} \mathbf{p}} \right).
		\end{equation}
	\end{enumerate}
\end{lemma}
\begin{proof}
	Part (i) is immediate. To establish the first formula in part (ii), we solve $\mathbf{B}(z^{*})\mathbf{w}=\mathbf{v}$
	for $\mathbf{p}/(\mathbf{q}^{\dag}\mathbf{p})$ and then use (i):
	\begin{equation}
		\frac{\mathbf{p}}{\mathbf{q}^{\dag}\mathbf{p}} = 
		\frac{z_{0} - z^{*}}{z_{0} - \zeta_{0}} \frac{\mathbf{w}}{\mathbf{q}^{\dag} \mathbf{w}}
		+ \frac{z^{*} - z_{0}}{z_{0} - \zeta_{0}} \frac{\mathbf{v}}{\mathbf{q}^{\dag} \mathbf{w}}
		= \frac{z_{0} - z^{*}}{z_{0} - \zeta_{0}} \frac{\mathbf{w}}{\mathbf{q}^{\dag} \mathbf{w}}
		+ \frac{z^{*} - \zeta_{0}}{z_{0} - \zeta_{0}} \frac{\mathbf{v}}{\mathbf{q}^{\dag} \mathbf{v}}.
	\end{equation}
	The second formula is proved in the similar way.
\end{proof}

Finally, we need the following notation.

\begin{notation} We define a \emph{twisting} of an elementary divisor $\mathbf{B}(z)$ by a constant
	non-degenerate matrix $\mathbf{A}$ to be a new elementary divisor ${^{\mathbf{A}}}\mathbf{B}(z)$
	such that ${^{\mathbf{A}}}\mathbf{B}(z) \mathbf{A} = \mathbf{A} \mathbf{B}(z)$, i.e., 
	\begin{equation}
		{^{\mathbf{A}}}\mathbf{B}(z) = \mathbf{A} \mathbf{B}(z) \mathbf{A}^{-1} = 
		\mathbf{I} + \frac{z_{0} - \zeta_{0}}{z - z_{0}} 
		\frac{(\mathbf{A} \mathbf{p})(\mathbf{q}^{\dag} \mathbf{A}^{-1})}{
		(\mathbf{q}^{\dag} \mathbf{A}^{-1}) (\mathbf{A} \mathbf{p})}.
	\end{equation}
\end{notation}


\subsection{Factors and Divisors}\label{ssec:factors_and_divisors} 

We begin with the following important remark. For the additive representation of $\mathbf{L}(z)$, the ordering
of zeroes and poles of $\det \mathbf{L}(z)$ is not important, but for any multiplicative representation 
choosing such an ordering is crucial. Thus, from now on our labeling will reflect the fact that 
$(\zeta_{s},z_{s})$-pair corresponds to some elementary divisor in $\mathbf{L}(z)$. 
There are two ways to look at the multiplicative structure of $\mathbf{L}(z)$ --- we can look at 
\emph{factors} or at \emph{divisors}.

\begin{definition}\label{def:factors} We say that elementary divisors 
	$\mathbf{C}_{s}(z) = \mathbf{I} + \mathbf{H}_{s}/(z-z_{s})$ corresponding to pairs 
	$(\zeta_{s},z_{s})$ with $\mathbf{H}_{s} = \mathbf{m}_{s} \mathbf{n}_{s}^{\dag}$ are the
	\emph{factors} of $\mathbf{L}(z)$ if 
	\begin{equation}
		\mathbf{L}(z) = \mathbf{L}_{0} \mathbf{C}_{1}(z) \cdots \mathbf{C}_{k}(z).
	\end{equation}	
\end{definition}

\begin{definition}\label{def:divisors} We say that elementary divisors $\mathbf{B}^{r}_{s}(z)$ (resp. 
	$\mathbf{B}^{l}_{s}(z)$) corresponding to pairs $(\zeta_{s},z_{s})$ are \emph{right} (resp. \emph{left})
	\emph{divisors} of $\mathbf{L}(z)$ if $\mathbf{L}(z) = \mathbf{L}^{r}_{s}(z) \mathbf{B}^{r}_{s}(z)$ 
	(resp. $\mathbf{L}(z) = \mathbf{B}^{l}_{s}(z)\mathbf{L}^{l}_{s}(z)$)
	where $\mathbf{L}^{r}_{s}(z)$ (resp. $\mathbf{L}^{l}_{s}(z)$) is regular at $z_{s}$.	
\end{definition}

The main advantage of the divisors as opposed to the factors is that they can be written explicitly in terms of 
the residues of $\mathbf{L}(z)$ and $\mathbf{M}(z)$. Note that $\mathbf{C}_{k}(z) = \mathbf{B}^{r}_{k}(z)$
and ${^{\mathbf{L}_{0}}}\mathbf{C}_{1}(z) = \mathbf{B}^{l}_{1}(z)$. In particular, in the quadratic case
$k=2$, there is no essential difference between divisors and factors.

\begin{lemma}\label{lem:divisors-formula}
	Let $\mathbf{L}_{s} = \operatorname{res}_{z_{s}} \mathbf{L}(z) = \mathbf{a}_{s} \mathbf{b}_{s}^{\dag}$
	and $\mathbf{M}_{s} = -\operatorname{res}_{\zeta_{s}} \mathbf{M}(z) = \mathbf{c}_{s} \mathbf{d}_{s}^{\dag}$.
	Then
	\begin{align}
		\mathbf{B}^{r}_{s}(z) &= \mathbf{I} + \frac{z_{s} - \zeta_{s}}{z - z_{s}} 
		\frac{\mathbf{c}_{s} \mathbf{b}_{s}^{\dag}}{\mathbf{b}_{s}^{\dag} \mathbf{c}_{s}} \label{eq:Br}\\
		\mathbf{B}^{l}_{s}(z) &= \mathbf{I} + \frac{z_{s} - \zeta_{s}}{z - z_{s}} 
		\frac{\mathbf{a}_{s} \mathbf{d}_{s}^{\dag}}{\mathbf{d}_{s}^{\dag} \mathbf{a}_{s}}. \label{eq:Bl}
	\end{align}
\end{lemma}
\begin{proof} To obtain the formula for the right divisor, we take the residue of
	$\mathbf{L}(z) = \mathbf{L}^{r}_{s}(z) \mathbf{B}^{r}_{s}(z)$ at $z_{s}$ to get 
	$\mathbf{b}_{s}^{\dag}\sim (\mathbf{q}^{r}_{s})^{\dag}$. Then we take the residue 
	of $\mathbf{M}(z) = (\mathbf{B}^{r}_{s}(z))^{-1} (\mathbf{L}^{r}_{s}(z))^{-1}$ at 
	$\zeta_{s}$ to get $\mathbf{c}_{s}\sim \mathbf{p}_{s}^{r}$. The formula then follows. 
	The expression for the left divisors is obtained in the same way.
\end{proof}

\subsection{Refactorization Transformations}\label{ssec:refactorization_transformations} 
As shown in \cite{Bor:2004:ITLSDE}, \cite{Kri:2004:ATDEWRECRP}, isomonodromic transformations on the 
space of rational matrix functions $\mathbf{L}(z)$ have the form 
\begin{equation}
	\mathbf{L}(z)\mapsto \tilde{\mathbf{L}}(z) = \mathbf{R}(z+1) \mathbf{L}(z) \mathbf{R}(z)^{-1},
\end{equation}
where $\mathbf{R}(z)$ is a certain rational matrix function. Similarly, we can consider
isospectral transformations
\begin{equation}
	\mathbf{L}(z)\mapsto \tilde{\mathbf{L}}(z) = \mathbf{R}(z) \mathbf{L}(z) \mathbf{R}(z)^{-1}.
\end{equation}
In the isospectral case we have to choose $\mathbf{R}(z)$ in such a way that the singularity structure of $\tilde{\mathbf{L}}(z)$ is 
the same as $\mathbf{L}(z)$, and in the isomonodromic case we require that $\mathbf{R}(z)$ induces integral shifts 
of certain parameters corresponding to the linear difference system given by $\mathbf{L}(z)$ 
(see \cite{Bor:2004:ITLSDE}, \cite{AriBor:2006:MSDDPE} for details). In particular, these requirements are satisfied if we take
$\mathbf{R}(z)$ to be one of the divisors of $\mathbf{L}(z)$, which corresponds to changing
the order of the factors in the multiplicative representation of $\mathbf{L}(z)$; in the isospectral case the divisor $\mathcal{D}$ is fixed and in
the isomonodromic case $\mathcal{D}$ is shifted by an integer vector. In what follows we only consider transformations of this particular 
type. We focus on the isospectral case, since the isomonodromic case is very similar.

Let us now restrict our attention to the quadratic case and take
$\mathbf{R}(z) = \mathbf{B}_{1}^{r}(z)$:
\begin{align}
	\mathbf{L}(z) &= \mathbf{B}_{2}^{l}(z) \mathbf{L}_{0} \mathbf{B}^{r}_{1}(z) \mapsto 
	\tilde{\mathbf{L}}(z)=\mathbf{B}^{r}_{1}(z) \mathbf{B}^{l}_{2}(z) \mathbf{L}_{0} = 
	\tilde{\mathbf{B}}_{2}^{l}(z) \mathbf{L}_{0} \tilde{\mathbf{B}}_{1}^{r}(z).
\end{align}

To study these transformation from the point of view of the discrete Euler-Lagrange equations it is necessary to 
construct a configuration space $\mathcal{Q}$, a Lagrangian function $\mathcal{L}\in \mathcal{F}(\mathcal{Q}\times \mathcal{Q})$,
and a map $\eta: \mathcal{Q}\times \mathcal{Q}\to \mathcal{M}_{r}^{\mathcal{D}}$ such that 
$\eta$ maps the Lagrangian dynamics to the refactorization dynamics. That is, denoting
by $\undertilde{\mathbf{Q}}$ the previous and by $\tilde{\mathbf{Q}}$ the next
point of the discrete dynamics, we want
$\eta(\undertilde{\mathbf{Q}},\mathbf{Q}) = \mathbf{L}(z)$ and 
$\eta(\mathbf{Q},\tilde{\mathbf{Q}}) = \tilde{\mathbf{L}}(z)$, where the shift map 
$\mathbf{\Phi}:(\undertilde{\mathbf{Q}},\mathbf{Q})\to (\mathbf{Q},\tilde{\mathbf{Q}})$ should satisfy the discrete Euler-Lagrange equations
\begin{equation}
	\frac{\partial \mathcal{L}}{\partial \mathbf{Y}}(\undertilde{\mathbf{Q}},\mathbf{Q}) + 
		\frac{\partial \mathcal{L}}{\partial \mathbf{X}}(\mathbf{Q},\tilde{\mathbf{Q}}) = 0,
\end{equation}
 see \cite{MosVes:1991:DVSCISFMP}, \cite{Ves:1991:ILRFMP}
for the general description of this approach, and \cite{Dzhamay:2008yq} for our specific situation.

Writing 
\begin{equation}
	\tilde{\mathbf{L}}(z) = \mathbf{B}^{r}_{1}(z) \mathbf{B}^{l}_{2}(z) \mathbf{L}_{0} = 
	\tilde{\mathbf{B}}_{1}^{l}(z) \mathbf{L}_{0} \tilde{\mathbf{B}}_{2}^{r}(z)
\end{equation}
and using the uniqueness of divisors, we see that $\mathbf{B}_{1}^{r}(z) = \tilde{\mathbf{B}}_{1}^{l}(z)$ and
$\mathbf{B}_{2}^{l}(z)\mathbf{L}_{0} = \mathbf{L}_{0} \tilde{\mathbf{B}}_{2}^{r}(z)$. From 
Lemma~\ref{lem:divisors-formula} it then follows that
$\mathbf{c}_{1} = \tilde{\mathbf{a}}_{1}$, $\mathbf{b}_{1}^{\dag} = \tilde{\mathbf{d}}_{1}^{\dag}$,
$\mathbf{a}_{2} = \mathbf{L}_{0} \tilde{\mathbf{c}}_{2}$, and $\mathbf{d}_{2}^{\dag} \mathbf{L}_{0} = \tilde{\mathbf{b}}_{2}^{\dag}$.
Since we want to parametrize $\mathbf{L}(z)$ by $(\undertilde{\mathbf{Q}},\mathbf{Q})$, we see that if we take 
$\mathbf{Q} = (\mathbf{a}_{2}, \mathbf{b}_{1}^{\dag})$ as one half of our coordinates, the second half  should be 
$\undertilde{\mathbf{Q}} = (\undertilde{\mathbf{a}_{2}}, \undertilde{\mathbf{b}_{1}^{\dag}}) 
= (\mathbf{L}_{0} \mathbf{c}_{2}, \mathbf{d}_{1}^{\dag})$. Now the proof of 
Theorem~\ref{thm:coords} is the same as the proof of Theorem~3.1 of \cite{Dzhamay:2008yq}, and it is sketched below.

\begin{proof}(Theorem~\ref{thm:coords})\label{pf:-coords}
	Taking the residue of 
	\begin{equation}
		\mathbf{L}(z) = \mathbf{B}_{1}^{l}(z) \mathbf{L}_{0} \mathbf{B}_{2}^{r}(z) = \mathbf{B}_{2}^{l}(z) \mathbf{L}_{0} \mathbf{B}_{1}^{r}(z)
	\end{equation}
	at the point $z_{1}$ and comparing the row spaces of the resulting rank-one matrices gives 
	$\mathbf{d}_{1}^{\dag} \mathbf{L}_{0} \mathbf{B}_{2}^{r}(z_{1}) = \mathbf{b}_{1}^{\dag}$. Using Lemma~\ref{lem:el-div-solve}(ii) we can 
	find the formula for $\mathbf{b}_{2}^{\dag}$,
	\begin{equation}
		\mathbf{b}_{2}^{\dag} = (z_{2} - z_{1}) \frac{\mathbf{d}_{1}^{\dag}\mathbf{L}_{0}}{\mathbf{d}_{1}^{\dag} \mathbf{L}_{0} \mathbf{c}_{2}} +
		(z_{1} - \zeta_{2})\frac{\mathbf{b}_{1}^{\dag}}{\mathbf{b}_{1}^{\dag} \mathbf{c}_{2}} = 
		\frac{\partial \mathcal{L}}{\partial \mathbf{x}_{2}} 
			((\mathbf{c}_{2},\mathbf{d}_{1}^{\dag}), (\mathbf{a}_{2}, \mathbf{b}_{1}^{\dag})),
	\end{equation}
	where $\mathcal{L}((\mathbf{x}_{2},\mathbf{x}_{1}^{\dag}),(\mathbf{y}_{2}, \mathbf{y}_{1}^{\dag}))$ is given by~(\ref{eq:Lagrangian}). 
	Performing similar calculations at the point $z_{2}$ for $\mathbf{L}(z)$ and the points $\zeta_{1}$, $\zeta_{2}$ for $\mathbf{M}(z)$ gives
	the rest of the formulas (\ref{eq:Ls}--\ref{eq:Ms}) and completes the proof.
\end{proof}

In these coordinates Theorem~3.1 of \cite{Dzhamay:2008yq} takes the following form (here we improve the
formulas from \cite{Dzhamay:2008yq} by using $\mathbf{L}_{0}$ instead of its root).
\begin{theorem}\label{thm:EL}
	\qquad{}
	\begin{enumerate}[(i)]
		\item The map $\eta: \mathcal{Q}\times \mathcal{Q}\to \mathcal{M}^{\mathcal{D}}_{r}$ is given by
		\begin{align}
			\eta(\undertilde{\mathbf{Q}},\mathbf{Q})= \mathbf{L}(z) 
			&=\left(\mathbf{I} + \frac{1}{z-z_{2}} \left((z_{2} - \zeta_{1}) 
						\frac{\mathbf{a}_{2} \undertilde{\mathbf{b}}_{1}^{\dag} }{ 
						\undertilde{\mathbf{b}}_{1}^{\dag} \mathbf{a}_{2} } + 
						(\zeta_{1} - \zeta_{2}) \frac{ \mathbf{a}_{2}\mathbf{b}_{1}^{\dag} \mathbf{L}_{0}^{-1}
						}{\mathbf{b}_{1}^{\dag} \mathbf{L}_{0}^{-1} \mathbf{a}_{2}}
			\right)	\right) \mathbf{L}_{0}\notag \\
			&\qquad\times \left(
			\mathbf{I} + \frac{1}{z-z_{1}}\left(
			(z_{1} - \zeta_{2})\frac{\mathbf{L}_{0}^{-1}\undertilde{\mathbf{a}}_{2} 
			\mathbf{b}_{1}^{\dag}}{\mathbf{b}_{1}^{\dag} \mathbf{L}_{0}^{-1}\undertilde{\mathbf{a}}_{2}}
			+ (\zeta_{2} - \zeta_{1})\frac{\mathbf{L}_{0}^{-1} \mathbf{a}_{2} \mathbf{b}_{2}^{\dag}}{
			\mathbf{b}_{2}^{\dag}\mathbf{L}_{0}^{-1} \mathbf{a}_{2} }	\right)	\right);
		\end{align}
		\item The equations of motion $(\mathbf{Q},\tilde{\mathbf{Q}}) = \mathbf{\Phi}(\undertilde{\mathbf{Q}},\mathbf{Q})$
		of both the isospectral and isomonodromic dynamics in these coordinates are given by the discrete Euler-Lagrange
		equations with the Lagrangian function
		\begin{align}
			\mathcal{L}(\mathbf{X},\mathbf{Y},t) & = 
			(z_{2} - z_{1}(t)) \log(\mathbf{x}_{1}^{\dag}  \mathbf{x}_{2}) + 
			(z_{1}(t) - \zeta_{2}) \log(\mathbf{y}_{1}^{\dag} \mathbf{L}_{0}^{-1}\mathbf{x}_{2})\notag \\
			&\qquad +(\zeta_{2} - \zeta_{1}(t))\log(\mathbf{y}_{1}^{\dag} \mathbf{L}_{0}^{-1} \mathbf{y}_{2}) +
			(\zeta_{1}(t) - z_{2}) \log(\mathbf{x}_{1}^{\dag} \mathbf{y}_{2}),			
		\end{align}	
		where in the isomonodromic case $z_{1}(t) = z_{1} - t$, $\zeta_{1}(t) = \zeta_{1} - t$, and in the isospectral case
		$z_{1}(t) = z_{1}$, $\zeta_{1}(t)=\zeta_{1}$ and $\mathcal{L}(\mathbf{X},\mathbf{Y})$ is time-independent.
	\end{enumerate}
\end{theorem}





\section{Isomonodromic Transformations and dPV} 
\label{sec:isomonodromic_transformations_and_dpv}

In \cite{AriBor:2006:MSDDPE}, Arinkin and Borodin showed that for rank-two matrices isomonodromic transformations above,
when written in a special coordinate system, are given by the difference Painlev\'e equations. In this section, choosing
the dPV case as an example, we show that these equations appear explicitly as relations between the residues of 
$\mathbf{L}^{\pm1}(z)$ and $\tilde{\mathbf{L}}^{\pm1}(z)$. Similar computation for the 
$q$-PVI case was done earlier by Jimbo and Sakai, \cite{JimSak:1996:AQOTSPE}.

\subsection{Spectral Coordinates}\label{ssec:spectral_coordinates} 
In the quadratic (two-pole) case the space of the rank-two matrices $\mathbf{L}(z)$ satisfying the requirements 
(\ref{L(z)-props-a}--\ref{L(z)-props-b}) can be described using different parameters. These parameters come in two groups.
The first group, that we call the \emph{type} of $\mathbf{L}(z)$, consists of the zeroes and poles of the determinant 
of $\mathbf{L}(z)$ and some asymptotic data at $z=\infty$. The space of $\mathbf{L}(z)$ of the fixed type is two-dimensional,
and the second group of parameters is a special coordinate system on this space, called the \emph{spectral coordinates}.
Expressing isomonodromy transformation in those coordinates gives rise to the 
difference Painlev\'e equations.

\begin{definition}\label{def:L-type}
	Let $\mathbf{L}(z)$ be a rational $2\times2$ matrix function on the Riemann sphere satisfying the 
	following conditions:
	\begin{align}
		\mathbf{L}(z) &= \operatorname{diag}\{\rho_{1},\rho_{2}\} + 
		\frac{\mathbf{L}_{1}}{z-z_{1}} + \frac{\mathbf{L}_{2}}{z-z_{2}},\quad
		\mathbf{L}_{i} = \mathbf{a}_{i} \mathbf{b}_{i}^{\dag},\quad
		\det \mathbf{L}(z) = \rho_{1} \rho_{2} \frac{(z-\zeta_{1})(z-\zeta_{2})}{(z-z_{1})(z-z_{2})}. \\
		\intertext{In addition, put}
		\rho_{1}k_{1} &= (\mathbf{L}_{\infty})_{11},\,
		\rho_{2}k_{2} = (\mathbf{L}_{\infty})_{22},\,
		\mu = (\mathbf{L}_{\infty})_{21}, \text{ where }
		\mathbf{L}_{\infty} = -\operatorname{res}_{\infty} \mathbf{L}(z)\,dz = \mathbf{L}_{1} + \mathbf{L}_{2}.
	\end{align}
	We call $(\rho_{1}, \rho_{2}, \zeta_{1}, \zeta_{2}, z_{1}, z_{2}, k_{1}, k_{2}, \mu)$ the \emph{type}
	of $\mathbf{L}(z)$. These parameters are not independent, since 
	\begin{equation}
		k_{1} + k_{2} = \operatorname{tr} \mathbf{L}_{0}^{-1} \mathbf{L}_{\infty} = (z_{1} - \zeta_{1}) + (z_{2} - \zeta_{2}).
	\end{equation}
	The conditions on $\rho_{i} k_{i}$ correspond to fixing the formal type of the solution of the difference equation at infinity,
	and the choice of $\mu$ corresponds to fixing the gauge under the global action by constant non-degenerate diagonal matrices.
\end{definition}

\begin{definition}\label{def:spectral-coordinates} The \emph{spectral coordinates} $(\gamma,\pi)$ are defined by the 
	conditions that 
	\begin{itemize}
		\item $\mathbf{L}(\gamma)_{21} = 0$ (and therefore $\mathbf{L}(z)_{21} = \frac{\mu (z-\gamma)}{(z-z_{1})(z-z_{2})}$);
		\item $\pi = \frac{(\gamma-z_{1})}{(\gamma-\zeta_{2})}\mathbf{L}(\gamma)_{11}$.
	\end{itemize}	
	The normalization conditions in this definition are chosen to match the formulas in \cite{AriBor:2006:MSDDPE} and
	\cite{AriBor:2007:TDITP}. 
\end{definition}	

\begin{notation} To find the expression of $\mathbf{L}(z)$ and $\mathbf{M}(z)$ in the spectral coordinates, it is convenient to 
introduce the notation 	$\varphi(a,b) =\pi (\gamma - a) -  \rho_{1}(\gamma - b)$.
\end{notation}

	\begin{lemma}\label{lem:L-spectral-coords-additive}
		Let $\mathbf{L}(z)$ be a rational $2\times2$ matrix function on the Riemann sphere that has the type
		$(\rho_{1}, \rho_{2}, \zeta_{1}, \zeta_{2}, z_{1}, z_{2}, k_{1}, k_{2}, \mu)$.
	Then, in spectral coordinates, the residues of the matrix $\mathbf{L}(z)$ are given by
	\begin{align}
		\mathbf{L}_{1} &= \mu \frac{\gamma - z_{1}}{z_{2} - z_{1}} \begin{bmatrix}
			\frac{1}{\mu}\left( \rho_{1}k_{1} - \frac{(\gamma - z_{2})}{\gamma - z_{1}}\varphi(\zeta_{2},z_{1})\right) \\ 1
		\end{bmatrix}
		\begin{bmatrix}
			1 & \frac{1}{\mu} \left( \rho_{2}k_{2}  + \frac{\rho_{2}}{\pi}\varphi(z_{2}, \zeta_{1}) \right)
		\end{bmatrix}, \label{eq:L1-spectral}\\
		\mathbf{L}_{2} &=  \mu \frac{\gamma - z_{2}}{z_{1} - z_{2}} \begin{bmatrix}
			\frac{1}{\mu}\left(\rho_{1} k_{1} - \varphi(\zeta_{2}, z_{1})\right) \\ 1
		\end{bmatrix} \begin{bmatrix}
			1 & \frac{1}{\mu} \left( \rho_{2} k_{2} + \frac{\rho_{2}(\gamma - z_{1})}{\pi(\gamma-z_{2})}\varphi(z_{2},\zeta_{1})\right)
		\end{bmatrix}.\label{eq:L2-spectral}
	\end{align}
\end{lemma}

\begin{proof} Let $\mathbf{L}_{i} = \alpha_{i} \begin{bmatrix}	a_{i} \\ 1\end{bmatrix} 
	\begin{bmatrix}	1 & b_{i}\end{bmatrix}$. Then 
	\begin{equation}
		\mathbf{L}(z)_{21}=\frac{\alpha_{1}}{z-z_{1}} + \frac{\alpha_{2}}{z-z_{2}} = \frac{\mu (z-\gamma)}{(z-z_{1})(z-z_{2})},
	\end{equation}
	and so $\alpha_{1} = \mu (\gamma - z_{1})/(z_{2} - z_{1})$, $\alpha_{2} = \mu (\gamma - z_{2})(z_{1} - z_{2})$, and 
	$\alpha_{1} + \alpha_{2} = \mu$.
	
	The normalization at infinity and the definition of $\pi$,
	\begin{align}
		\rho_{1} k_{1} &= \alpha_{1} a_{1} + \alpha_{2} a_{2} = \mu a_{1} + \alpha_{2}(a_{2} - a_{1}) = \alpha_{1}(a_{1} - a_{2}) + \mu a_{2},\\
		\mathbf{L}(\gamma)_{11} &= \pi\frac{(\gamma - \zeta_{2})}{(\gamma - z_{1})} = \rho_{1} + \frac{\alpha_{1}(a_{1} - a_{2})}{(\gamma - z_{1})}
		= \rho_{1} + \frac{\alpha_{2}(a_{2} - a_{1})}{(\gamma - z_{2})},
		\intertext{immediately give}
		a_{1} &= \frac{1}{\mu}(\rho_{1} k_{1} - \alpha_{2}(a_{2} - a_{1})) 
		= \frac{1}{\mu}\left( \rho_{1} k_{1} - \frac{(\gamma - z_{2})}{(\gamma - z_{1})} \varphi(\zeta_{2}, z_{1})\right),\\
		a_{2} &= \frac{1}{\mu}\left(\rho_{1} k_{1} - \varphi(\zeta_{2}, z_{1})\right).
	\end{align}
	
	Using the equation $\mathbf{L}(\gamma)_{11} \mathbf{L}(\gamma)_{22} = \det \mathbf{L}(\gamma)$	we get 
	$\mathbf{L}(\gamma)_{22} = \frac{\rho_{1} \rho_{2}}{\pi}\frac{(\gamma - \zeta_{1})}{(\gamma - z_{2})}$. This, and the condition 
	$\rho_{2} k_{2} = \alpha_{1} b_{1} + \alpha_{2} b_{2}$, allows us to find the expressions for $b_{1}$, $b_{2}$ in exactly the same way.
\end{proof}

\begin{corollary}\label{cor:M-spectral-coords-additive} In the same gauge, the residues $\mathbf{M}_{i}$ of the inverse matrix
\begin{equation}
	\mathbf{M}(z) = \mathbf{L}(z)^{-1} = \operatorname{diag}\{1/\rho_{1},1/\rho_{2}\} - \frac{\mathbf{M}_{1}}{z - \zeta_{1}} 
	 - \frac{\mathbf{M}_{2}}{z - \zeta_{2}}
\end{equation}	
are given by 	
\begin{align}
	\mathbf{M}_{1} &= \frac{\mu}{\rho_{1}\rho_{2}} \frac{\gamma - \zeta_{1}}{\zeta_{2} - \zeta_{1}} \begin{bmatrix}
		\frac{1}{\mu}\left(\rho_{2} k_{1} - \frac{\rho_{2}}{\pi}\varphi(\zeta_{2},z_{1}) \right)\\ 1
	\end{bmatrix}\begin{bmatrix}
		1 & \frac{1}{\mu}\left(\rho_{1} k_{2} + \frac{\gamma - \zeta_{2}}{\gamma - \zeta_{1}} \varphi(z_{2},\zeta_{1})\right)
	\end{bmatrix}, \label{eq:M1-spectral}\\
	\mathbf{M}_{2} &= \frac{\mu}{\rho_{1} \rho_{2}} \frac{\gamma - \zeta_{2}}{\zeta_{1} - \zeta_{2}} \begin{bmatrix}
		\frac{1}{\mu}\left(\rho_{2} k_{1} - \frac{\rho_{2} (\gamma - \zeta_{1})}{\pi(\gamma - \zeta_{2})}\varphi(\zeta_{2},z_{1}) \right)\\ 1
	\end{bmatrix}\begin{bmatrix}
		1 & \frac{1}{\mu}\left(\rho_{1} k_{2} + \varphi(z_{2},\zeta_{1})\right)
	\end{bmatrix}. \label{eq:M2-spectral}
\end{align}
\end{corollary}	

\begin{proof}
	Since $\mathbf{M}(z)$ has the same form as $\mathbf{L}(z)$, we only have to determine the type
	and spectral coordinates of $\mathbf{M}(z)$ in terms of those of $\mathbf{L}(z)$, 
	\begin{center}
		\begin{tabular}{rccccccccccc}
		$\mathbf{L}(z)$: & \quad $z_{1}$ & $z_{2}$ & $\zeta_{1}$ & $\zeta_{2}$ & $\rho_{1}$ & $\rho_{2}$ & $k_{1}$ & $k_{2}$ & $\mu$ & 
		$\gamma$ & $\pi$ \\
		$\mathbf{M}(z)$: & \quad $\zeta_{1}$ & $\zeta_{2}$ & $z_{1}$ & $z_{2}$ & $\frac{1}{\rho_{1}}$ & $\frac{1}{\rho_{2}}$ & $-k_{1}$ & $-k_{2}$ 
		& $-\frac{\mu}{\rho_{1}\rho_{2}}$ & 
		$\gamma$ & $\frac{(\gamma-z_{1})(\gamma-\zeta_{1})}{\pi(\gamma-z_{2})(\gamma-\zeta_{2})}$ 		
		\end{tabular}		
	\end{center}
	and take the negative sign in the definitions of $\mathbf{M}_{i}$ into account. Note that in computing the type of $\mathbf{M}(z)$ we used the
	equation $\mathbf{M}_{\infty} = - \mathbf{L}_{0}^{-1} \mathbf{L}_{\infty} \mathbf{L}_{0}^{-1}$, which follows from the condition 
	$\operatorname{res}_{\infty}\mathbf{L}(z) \mathbf{M}(z) = \mathbf{0}$.
\end{proof}


\subsection{Difference Painlev\'e V}\label{ssec:difference_painlev_e_v} 
Consider now the isomonodromy transformation given by $\mathbf{R}(z) = \mathbf{B}^{r}_{1}(z)$:
\begin{equation}
	\mathbf{L}(z) = \mathbf{B}^{l}_{2}(z) \mathbf{L}_{0} \mathbf{B}^{r}_{1}(z) \mapsto 
	\tilde{\mathbf{L}}(z) = \mathbf{B}^{r}_{1}(z+1) \mathbf{B}^{l}_{2}(z) \mathbf{L}_{0} = 
	\tilde{\mathbf{B}}^{l}_{1}(z) \mathbf{L}_{0} \tilde{\mathbf{B}}^{r}_{2}(z).
\end{equation}

\begin{theorem}\label{thm:iso-dpV}
	The type and spectral coordinates of $\tilde{\mathbf{L}}(z)$ in terms of those of $\mathbf{L}(z)$ are given by
	\begin{equation}
			\begin{tabular}{rccccccccccc}
			$\mathbf{L}(z)$: & \quad $z_{1}$ & $z_{2}$ & $\zeta_{1}$ & $\zeta_{2}$ & $\rho_{1}$ & $\rho_{2}$ & $k_{1}$ & $k_{2}$ & $\mu$ & 
			$\gamma$ & $\pi$ \\
			$\tilde{\mathbf{L}}(z)$: & \quad $\tilde{z}_{1}=z_{1}-1$ & $\tilde{z}_{2}=z_{2}$ & $\tilde{\zeta}_{1}=\zeta_{1}-1$ & 
			$\tilde{\zeta}_{2}=\zeta_{2}$ & $\rho_{1}$ & $\rho_{2}$ & $k_{1}$ & $k_{2}$ & 
			$\tilde{\mu}$ & 
			$\tilde{\gamma}$ & $\tilde{\pi}$ 		
			\end{tabular},		
	\end{equation}
	where 
	\begin{align}
		\tilde{\mu} &= \mu \frac{\rho_{1}(\pi - \rho_{2})}{\rho_{2}(\pi - \rho_{1})} \label{eq:dPV-0}\\
		\tilde{\gamma} + \gamma & = z_{2} + \zeta_{2} + \frac{\rho_{1}(k_{1} - z_{1} + \zeta_{2})}{\pi-\rho_{1}} + 
		\frac{\rho_{2}(k_{2} - z_{1} + \zeta_{2} + 1)}{\pi - \rho_{2}}\label{eq:dPV-1}\\
		\tilde{\pi}\pi &= \rho_{1} \rho_{2}  \frac{(\tilde{\gamma} - \tilde{z}_{1}) (\tilde{\gamma} - \tilde{\zeta}_{1})}{
		(\tilde{\gamma} - \tilde{z}_{2}) (\tilde{\gamma} - \tilde{\zeta}_{2})} \label{eq:dPV-2}
	\end{align}
	Equations (\ref{eq:dPV-1}--\ref{eq:dPV-2}) are the difference Painlev\'e V equations of Sakai's hierarchy \cite{Sak:2001:RSAWARSGPE}, first obtained in this setting
	in \cite{AriBor:2006:MSDDPE} (Theorem B).
\end{theorem}

\begin{proof}
	First note that 
	\begin{equation}
		\mathbf{L}_{\infty} = \mathbf{G}_{2}^{l} \mathbf{L}_{0} + \mathbf{L}_{0}\mathbf{G}_{1}^{r},\qquad
		\tilde{\mathbf{L}}_{\infty} = (\mathbf{G}_{1}^{r} + \mathbf{G}_{2}^{l}) \mathbf{L}_{0}.
	\end{equation}
	Thus, using (\ref{eq:Br}), (\ref{eq:L1-spectral}), and (\ref{eq:M1-spectral}), we get
	\begin{align*}
		\tilde{\mu} &= (\tilde{\mathbf{L}}_{\infty})_{21} = \mu + [\mathbf{G}^{r}_{1},\mathbf{L}_{0}]_{21}
		= \mu + (\rho_{1}- \rho_{2}) (\mathbf{G}^{r}_{1})_{21} = 
		\mu + (\rho_{1} - \rho_{2})\frac{z_{1} - \zeta_{1}}{\mathbf{b}_{1}^{\dag} \mathbf{c}_{1}}	\\
		&= \mu\left(1 + \frac{\pi(\rho_{1} - \rho_{2})}{\rho_{2}(\pi - \rho_{1})}\right)  = \mu \frac{\rho_{1}(\pi - \rho_{2})}{\rho_{2}(\pi - \rho_{1})}.
	\end{align*}
	
	From the uniqueness of the left and right divisors we see that $\mathbf{B}_{1}^{r}(z+1) = \tilde{\mathbf{B}}_{1}^{l}(z)$ and 
	$\mathbf{B}_{2}^{l}(z) = {^{\mathbf{L}_{0}}} \tilde{\mathbf{B}}_{2}^{r}(z)$. Using the first equation, we see that 
	$\mathbf{G}_{1}^{r} = \tilde{\mathbf{G}}_{1}^{l}$ and so
	$(\mathbf{c}_{1}\mathbf{b}_{1}^{\dag})/(\mathbf{b}_{1}^{\dag} \mathbf{c}_{1}) 
	= (\tilde{\mathbf{a}}_{1} \tilde{\mathbf{d}}_{1}^{\dag})/(\tilde{\mathbf{d}}_{1}^{\dag} \tilde{\mathbf{a}}_{1})$. In particular,
	\begin{equation}
		\frac{\tilde{\mu}}{\mu} \mathbf{b}_{1}^{\dag} \mathbf{c}_{1} = (z_{1} - \zeta_{1}) \frac{\rho_{1}(\pi-\rho_{2})}{\pi} = 
		\tilde{\mathbf{d}}_{1}^{\dag} \tilde{\mathbf{a}}_{1} = 
		(\tilde{z}_{1} - \tilde{\zeta}_{1})\left(\rho_{1} - 
		\frac{\tilde{\pi} (\tilde{\gamma} - \tilde{z}_{2}) (\tilde{\gamma} - \tilde{\zeta}_{2})}{
		(\tilde{\gamma} - \tilde{z}_{1}) (\tilde{\gamma} - \tilde{\zeta}_{1})
		}\right),
	\end{equation}
	which gives (\ref{eq:dPV-2}). Similarly, comparing the first components of the normalized vectors  $\mathbf{c}_{1}$ and $\tilde{\mathbf{a}}_{1}$ gives
	(\ref{eq:dPV-1}) and completes the proof.
\end{proof}

\thanks{
I am very grateful to I.~Krichever, A.~Borodin, M.~Gekhtman, H.~Sakai, F.~Soloviev  and T.~Takenawa for interesting and helpful discussions. I also want to thank the organizers 
of the \textbf{SIDE 8} conference for the invitation to participate at the conference and the opportunity to present a part of this work. Finally, I thank the referees for
many useful suggestions.}

\section*{References}
\bibliographystyle{unsrt}

\end{document}